\documentclass[12pt]{article}
\usepackage[utf8]{inputenc}

    \usepackage[dvipsnames]{xcolor}
    \usepackage[T1]{fontenc}
	\usepackage[utf8]{inputenc}
	\usepackage[english]{babel}
	\usepackage{amsfonts}
	\usepackage{amssymb}
	    \usepackage{amsmath}
	\usepackage{threeparttablex}
\usepackage{threeparttable}
	\usepackage{enumerate}
	\usepackage{bbm}
	\usepackage{mathtools}
	\usepackage{subfig}
	\usepackage{booktabs}	
	\usepackage[nodisplayskipstretch]{setspace}
	 
\allowdisplaybreaks
	\usepackage{graphicx}
	\usepackage{adjustbox}
	\usepackage{fullpage}
	
	\usepackage{soul}
	\usepackage{xspace}
	\usepackage{csquotes}
	
	\usepackage[margin = 1in]{geometry}
	\usepackage{placeins}
	\usepackage[footfont=normalsize,font=normalsize]{floatrow}
	\usepackage{lscape}
	\usepackage{longtable}
	\usepackage{mathabx}
    \usepackage{accents}
    \usepackage{float}

    \usepackage{tikz}
    \usetikzlibrary{calc,intersections}
    
	\usepackage[longnamesfirst]{natbib}
	\usepackage{bibunits}
	\defaultbibliography{Bibliography.bib}  
	\defaultbibliographystyle{aer}

	\usepackage{hyperref}
	\hypersetup{
    colorlinks=true,
    linkcolor=Red,
    filecolor=magenta,      
    urlcolor=blue,
    citecolor = Blue
}
	\usepackage{cleveref}

    \usepackage{rotating}
	\usepackage[framemethod=tikz]{mdframed}

    \usepackage{tcolorbox}
    \newtcbox{\feedback}{nobeforeafter,colframe=black,colback=white,boxrule=0.5pt,arc=2pt,
      boxsep=0pt,left=2pt,right=2pt,top=2pt,bottom=2pt,tcbox raise base}
      
    \usepackage{amsthm}

    \newtheorem{prop}{Proposition}[section]
    \newtheorem*{prop*}{Proposition}
    \newtheorem{lem}{Lemma}[section]
    
    \theoremstyle{definition}
    \newtheorem{rem}{Remark}
    \newtheorem{defn}{Definition}
    \newtheorem{example}{Example}

\usepackage{array}
\newcolumntype{L}[1]{>{\raggedright\let\newline\\\arraybackslash}m{#1}}
\newcolumntype{C}[1]{>{\centering\let\newline\\\arraybackslash\hspace{0pt}}m{#1}}
\newcolumntype{R}[1]{>{\raggedleft\let\newline\\\arraybackslash\hspace{0pt}}m{#1}}

\usepackage{ragged2e}
\newlength\ubwidth

\newcommand\Item[1][]{%
  \ifx\relax#1\relax  \item \else \item[#1] \fi
  \abovedisplayskip=0pt\abovedisplayshortskip=0pt~\vspace*{-\baselineskip}}

	\newcommand{\bracks}[1]{\left[#1\right]}
	\newcommand{\expe}[1]{\mathbb{E}\bracks{#1}}

	\newcommand{\reals}{\mathbb{R}}

\newcommand{\ytilde}{\tilde{y}}

\hypersetup{final}

\AtBeginDocument{
\addtolength{\abovedisplayskip}{-1ex}
\addtolength{\belowdisplayskip}{-1ex}
}

\allowdisplaybreaks
\author{Jonathan Roth\thanks{Brown University. Email: \href{mailto:jonathanroth@brown.edu}{jonathanroth@brown.edu}} \and Pedro H.C. Sant'Anna\thanks{Vanderbilt University and Microsoft. Email: \href{mailto:pedro.h.santanna@vanderbilt.edu}{pedro.h.santanna@vanderbilt.edu}} }
\date{\today}
\title{When Is Parallel Trends Sensitive to Functional Form?\thanks{We thank Isaiah Andrews, Otavio Bartalotti, Kirill Borusyak, Kevin Chen, Carol Caetano, Arin Dube, Dalia Ghanem, Andrew Goodman-Bacon, Ryan Hill, Martin Huber, Peter Hull, Guido Imbens, Ariella Kahn-Lang, Kevin Lang, Arthur Lewbel, Daniel Millimet, \'{A}ureo de Paula, Ashesh Rambachan, Adrienne Sabety, Yuya Sasaki, Jesse Shapiro, Tymon S\l{}oczy\'{n}ski, Alex Torgovitsky, Kaspar W\"uthrich, three anonymous referees, and seminar participants at Brandeis, University of Chicago, Hebrew University, LMU Munich, UC-Berkeley, UC-Davis, the Southern Economic Association annual conference, and the Chamberlain Seminar for helpful comments and conversations.}}

\begin{document}
\begin{bibunit}

\maketitle


\begin{abstract}
This paper assesses when the validity of difference-in-differences depends on functional form. We provide a novel characterization: the parallel trends assumption holds under all strictly monotonic transformations of the outcome if and only if a stronger ``parallel trends''-type condition holds for the cumulative distribution function of untreated potential outcomes. This condition for parallel trends to be insensitive to functional form is satisfied if and essentially only if the population can be partitioned into a subgroup for which treatment is effectively randomly assigned and a remaining subgroup for which the distribution of untreated potential outcomes is stable over time. These conditions have testable implications, and we introduce falsification tests for the null that parallel trends is insensitive to functional form.
\end{abstract}

\newpage

\section{Introduction}
This paper studies when the parallel trends assumption used for identification in difference-in-differences (DiD) designs is insensitive to functional form. The motivation for studying this property is that it often may not be clear from theory that parallel trends should hold for a particular choice of functional form. For example, a researcher may be interested in the average treatment effect on the treated (ATT) in levels for a particular policy, but it may not be obvious that state-level variation in the policy generates parallel trends specifically in levels rather than in logs or some other transformation. The DiD design thus will often be more credible if its validity does not depend on functional form. Our results make precise the conditions needed to ensure this form of robustness, and suggest that researchers should be careful to give a functional form-specific justification in settings where these conditions are not plausible.

We provide two characterizations of when parallel trends is insensitive to functional form, in the sense that it holds for all strictly monotonic transformations of the outcome. First, we show that parallel trends is insensitive to functional form if and only if a ``parallel trends''-type condition holds for the entire cumulative distribution function (CDF) of $Y(0)$. We further show that this condition can be satisfied if and essentially only if we are in one of three cases: (i) when treatment is as-if randomly assigned, (ii) when the distribution of $Y(0)$ is stable over time for each treatment group; and (iii) a hybrid of the first two cases in which the population is a mixture of a sub-population that is (as-if) randomized between treatment and control and another sub-population that has stable untreated potential outcome distributions over time. In settings where the treatment is not (as-if) randomly assigned, the assumptions needed for the insensitivity of parallel trends to functional form will thus often be quite restrictive.

These conditions have testable implications, and we introduce tests for the null hypothesis that parallel trends is insensitive to functional form. Such tests can be useful in flagging situations where the researcher should be particularly careful about justifying the parallel trends assumption for the specific functional form chosen for the analysis. We illustrate how the proposed tests can be used in a stylized analysis of the effects of the minimum wage.


Previous papers have noted that the parallel trends assumption may hold in logs but not levels or vice versa \citep[e.g.][]{meyer_natural_1995,AtheyImbens(06), kahn-lang_promise_2020}. However, to our knowledge we provide the first full characterization of when it is sensitive to functional form. The conditions that we derive are also related to, but distinct from, assumptions introduced previously for identifying distributional treatment effects in DiD settings \citep[e.g.,][]{AtheyImbens(06)}; see Remark \ref{rem: cic} for details.

\section{Model\label{sec: model}}

We consider a canonical two-period difference-in-differences model. There are two periods $t=0,1$, and units indexed by $i$ come from one of two populations denoted by $D_i \in \{0,1\}$. Units in the $D_i=1$ (treated) population receive treatment beginning in period $t=1$, and units in the $D_i=0$ (comparison) population never receive treatment. We denote by $Y_{it}(1), Y_{it}(0)$ the potential outcomes for unit $i$ in period $t$ under treatment and control, respectively, and we observe the outcome $Y_{it} = D_i Y_{it}(1) + (1-D_i) Y_{it}(0)$. We assume that there are no anticipatory effects of treatment, so that $Y_{i0}(1) = Y_{i0}(0)$ for all $i$. The average treatment effect on the treated is defined as
$$\tau_{ATT} = \expe{Y_{i1}(1) - Y_{i1}(0) \,|\, D_i = 1}.$$


\begin{rem}[Implications for other settings]
We consider a two-period model for expositional simplicity. Our results have immediate implications for DiD settings with multiple periods and staggered treatment timing, as in these contexts researchers often impose that the two-group, two-period version of parallel trends holds for many combinations of groups and time periods; see \citet{roth_et_al_DiD_survey} for a review. Similarly, our identification results would go through in settings with conditional parallel trends \citep[e.g.,][]{abadie_semiparametric_2005, santanna_doubly_2020} if all probability statements were conditional on covariates. 
\end{rem}



\section{Invariance of Parallel Trends\label{sec: did}}

The classical assumption that allows for point identification of the ATT in the DiD design is the parallel trends assumption, which imposes that
\begin{equation}
\expe{ Y_{i1}(0) \,|\, D_i =1} - \expe{ Y_{i0}(0) \,|\, D_i = 1 } = \expe{ Y_{i1}(0) \,|\, D_i =0}  - \expe{ Y_{i0}(0) \,|\, D_i = 0}. \label{eqn: parallel trends assumption - superpop}    
\end{equation}

\noindent Under the parallel trends assumption, the ATT is identified: $\tau_{ATT} = (\mu_{11} - \mu_{10}) - (\mu_{01} - \mu_{00})$, where $\mu_{dt} = \expe{Y_{it} \,|\, D_i=d}$. We assume throughout that the four expectations in (\ref{eqn: parallel trends assumption - superpop}) exist and are finite. Following \citet{AtheyImbens(06)}, we say that the parallel trends assumption is invariant to transformations if the parallel trends assumption holds for all strictly monotonic transformations of the outcome.

\begin{defn}\label{def: invariance to scale}
We say that the parallel trends assumption is invariant to transformations (a.k.a. insensitive to functional form)  if 
\begin{multline*}
\expe{ g(Y_{i1}(0)) \,|\, D_i =1} - \expe{ g(Y_{i0}(0)) \,|\, D_i = 1 } 
= \expe{ g(Y_{i1}(0)) \,|\, D_i= 0} - \expe{ g(Y_{i0}(0)) \,|\, D_i = 0}
\end{multline*}
\noindent for all strictly monotonic functions $g$ such that the expectations above are finite.
\end{defn}

Our first result characterizes when parallel trends is invariant to transformations.

\begin{prop} \label{prop: invariance to transformation did - superpop}
Parallel trends is invariant to transformations if and only if
\begin{equation}
F_{Y_{i1}(0) | D_i=1}(y) - F_{Y_{i0}(0) | D_i=1}(y) = F_{Y_{i1}(0) | D_i=0}(y) - F_{Y_{i0}(0) | D_i=0}(y), \text{ for all } y\in\reals \label{eqn: parallel trends for cdfs - superpop}\end{equation}
\noindent where $F_{Y_{it}(0) | D_i = d}(y)$ is the cumulative distribution function of $Y_{it}(0) \,|\, D_i = d$. 
\end{prop}

\begin{proof}
If (\ref{eqn: parallel trends for cdfs - superpop}) holds, then from integrating on both sides of the equation it is immediate that 
\begin{equation}
\int g(y) d F_{Y_{i1}(0)|D_i=1} - \int g(y) d F_{Y_{i0}(0)|D_i=1} = \int g(y) d F_{Y_{i1}(0)|D_i=0} - \int g(y) d F_{Y_{i0}(0)|D_i=0} \label{eqn: integral formulation - invariance to transformations} \end{equation}
\noindent for any strictly monotonic $g$ such that the integrals exist and are finite, and hence parallel trends is invariant to transformations. 

Conversely, if parallel trends is invariant to transformations, then (\ref{eqn: integral formulation - invariance to transformations}) holds for every strictly monotonic $g$ such that the expectations exist and are finite. In particular, it holds for the identity map $g_1(y) = y$ as well as the map $g_2(y) = y - 1[y\leq \ytilde]$ for any given $\ytilde \in \reals$. Then, it follows that 
\begin{align*}
& \int y d F_{Y_{i1}(0)|D_i=1} - \int y d F_{Y_{i0}(0)|D_i=1} = \int y d F_{Y_{i1}(0)|D_i=0} - \int y d F_{Y_{i0}(0)|D_i=0},  \text{ and } \\
& \int (y - 1[y\leq\ytilde]) d F_{Y_{i1}(0)|D_i=1} - \int (y - 1[y\leq\ytilde]) d F_{Y_{i0}(0)|D_i=1} = \\
&\int (y - 1[y\leq\ytilde]) d F_{Y_{i1}(0)|D_i=0} - \int (y - 1[y\leq\ytilde]) d F_{Y_{i0}(0)|D_i=0}.
\end{align*}
\noindent Subtracting the second equation from the first in the previous display, we obtain
\begin{small}
\begin{equation*}
\int 1[y\leq\ytilde] d F_{Y_{i1}(0)|D_i=1} - \int  1[y\leq\ytilde] d F_{Y_{i0}(0)|D_i=1} = \int 1[y\leq\ytilde] d F_{Y_{i1}(0)|D_i=0} - \int 1[y\leq\ytilde] d F_{Y_{i0}(0)|D_i=0},   
\end{equation*}
\end{small}
\noindent which is equivalent to (\ref{eqn: parallel trends for cdfs - superpop}) by the definition of the CDF and the fact that $\ytilde$ is arbitrary.
\end{proof}

\noindent Proposition \ref{prop: invariance to transformation did - superpop} shows that parallel trends is invariant to transformations if and only if a ``parallel trends''-type assumption holds for the CDFs of the untreated potential outcomes. We note that if the outcome is continuous, then parallel trends of CDFs is equivalent to parallel trends of PDFs (almost everywhere). The following result provides a characterization of how distributions satisfying this assumption can be generated. 

\begin{prop}\label{prop: generation of parallel trends of cdfs}
Suppose that the distributions $Y_{it}(0)|D_i=d$ for all $d,t \in \{0,1\}$ have a Radon-Nikodym density with respect to a common dominating, positive $\sigma$-finite measure.
Then parallel trends is invariant to transformations if and only if there exists $\theta \in [0,1]$ and CDFs $G_{t}(\cdot)$ and $H_{d}(\cdot)$ depending only on time and group, respectively, such that
\begin{equation}
F_{Y_{it}(0) | D_i = d}(y) = \theta G_{t}(y) + (1-\theta) H_{d}(y) \text{ for all } y\in \reals \text{ and }  d,t \in \{ 0,1 \}. \label{eqn: decomp of cdfs into mixtures}
\end{equation}
\end{prop}

\begin{proof}
See Appendix \ref{sec: appendix proof of decomp of cdfs into mixtures}.
\end{proof}
\noindent Proposition \ref{prop: generation of parallel trends of cdfs} shows that parallel trends of CDFs is satisfied if and only if the untreated potential outcomes for each group and time can be represented as a mixture of a common time-varying distribution that does not depend on group (with weight $\theta$) and a group-specific distribution that does not depend on time (with weight $1-\theta$). This implies that there are three cases in which parallel trends will be insensitive to functional form, depending on the value of $\theta$. 

\paragraph{Case 1: Random Assignment. $\boldsymbol{(\theta =1)}$.} The case $\theta =1$ corresponds with imposing that the distributions of $Y(0)$ for the treated and comparison groups are the same in each period, $F_{Y_{it}(0) | D_i = 1}(y) = F_{Y_{it}(0) | D_i = 0}(y)$, for $t=0,1$ and all $y$, as occurs under (as-if) random assignment of treatment.

\paragraph{Case 2: Stationary $\boldsymbol{Y(0)}$. $\boldsymbol{(\theta =0})$.} The case $\theta =0$ corresponds with imposing that the distribution of $Y(0)$ for both the treated and comparison populations does not depend on time, i.e. $F_{Y_{i1}(0) | D_i = d}(y) = F_{Y_{i0}(0) | D_i = d}(y)$, for $d=0,1$ and all $y$.

\paragraph{Case 3: Non-random assignment and non-stationarity. $\boldsymbol{(\theta \in (0,1)})$.} The case $\theta \in (0,1)$ corresponds with a hybrid of Cases 1 and 2. In each period, we can partition the treated and comparison groups so that $\theta$ fraction of each group have the same distribution $(G_t)$, as if they were randomly assigned, and $1-\theta$ fraction have a group-specific distribution that does not depend on time ($H_d$). This is perhaps the most interesting case, since in the other two cases a single difference (either across time or across groups) would suffice to identify the ATT.

\begin{rem}[Use of phrase ``essentially only if''] \label{rem: essentially only if} The simplest way for Case 3 to hold is to have a $\theta$ fraction of the population that is as-if randomized between treatment and control and a $1-\theta$ fraction that has stationary potential outcomes. In principle, though, it is possible for the units in the $\theta$ and $1-\theta$ partitions to change across periods in Case 3, although it is difficult to imagine scenarios where this would be the case in practice. We thus write in the abstract that parallel trends can be invariant to transformations ``essentially only if'' the population can be partitioned into groups such that one is effectively randomized between treatment and control, and the others have stable potential outcomes over time.
\end{rem}


\begin{example}[Binary outcomes]
Suppose the outcome is binary, $Y_i \in \{0,1\}$. Then for any $y \in [0,1)$, $F_{Y_{i1}(0) | D_i = 1}(y) = 1 - \expe{ Y_{i1}(0) \,|\, D_i = 1}$, and analogously for the other CDFs. Thus, (\ref{eqn: parallel trends for cdfs - superpop}) is equivalent to the parallel trends assumption (\ref{eqn: parallel trends assumption - superpop}). Proposition \ref{prop: invariance to transformation did - superpop} thus implies that whenever the parallel trends assumption holds, it also holds for all monotonic transformations of the outcome (i.e. replacing $\{0,1\}$ with $\{a,b\}$). This is intuitive, as the expectation of a binary outcome fully characterizes its distribution.
\end{example}

\begin{example}[Normally distributed outcomes] \label{example: normal example}
If the distribution of $Y_{it}(0) | D_i = d$ is normally distributed with positive variance for all $(d,t)$, then (\ref{eqn: parallel trends for cdfs - superpop}) can hold only if either both groups have the same distribution of $Y(0)$ in each period or the group-specific distributions of $Y(0)$ do not change over time (Cases 1 and 2).
\end{example}

\begin{example}[Mixtures of distributions]\label{ex:case3}

An example of Case 3 is illustrated in Figure \ref{fig:case3_example}. The distributions of potential outcomes are generated by equation (\ref{eqn: decomp of cdfs into mixtures}), with $\theta = \frac{1}{2}$, $G_t \sim lognormal(2+t,1)$, and $H_d \sim lognormal(3+d,1)$. As can be seen in the figure, the distributions of $Y(0)$ for the treatment and comparison groups differ from each other in both time periods and change over time. However, the change in the PDFs is the same for both groups, and thus our results imply that parallel trends holds for all monotonic transformations of the outcome, as illustrated in Table \ref{tab:case3} for the levels and log transformations. These distributions of $Y(0)$ could arise, for example, if half the population is younger workers, who have untreated earnings $G_t$ in period $t$ regardless of treatment group, and the other half is older workers who have untreated earnings $H_d$ in both periods.
\end{example}
\begin{figure}[!htp]
    \centering
    \includegraphics[width = 1\linewidth]{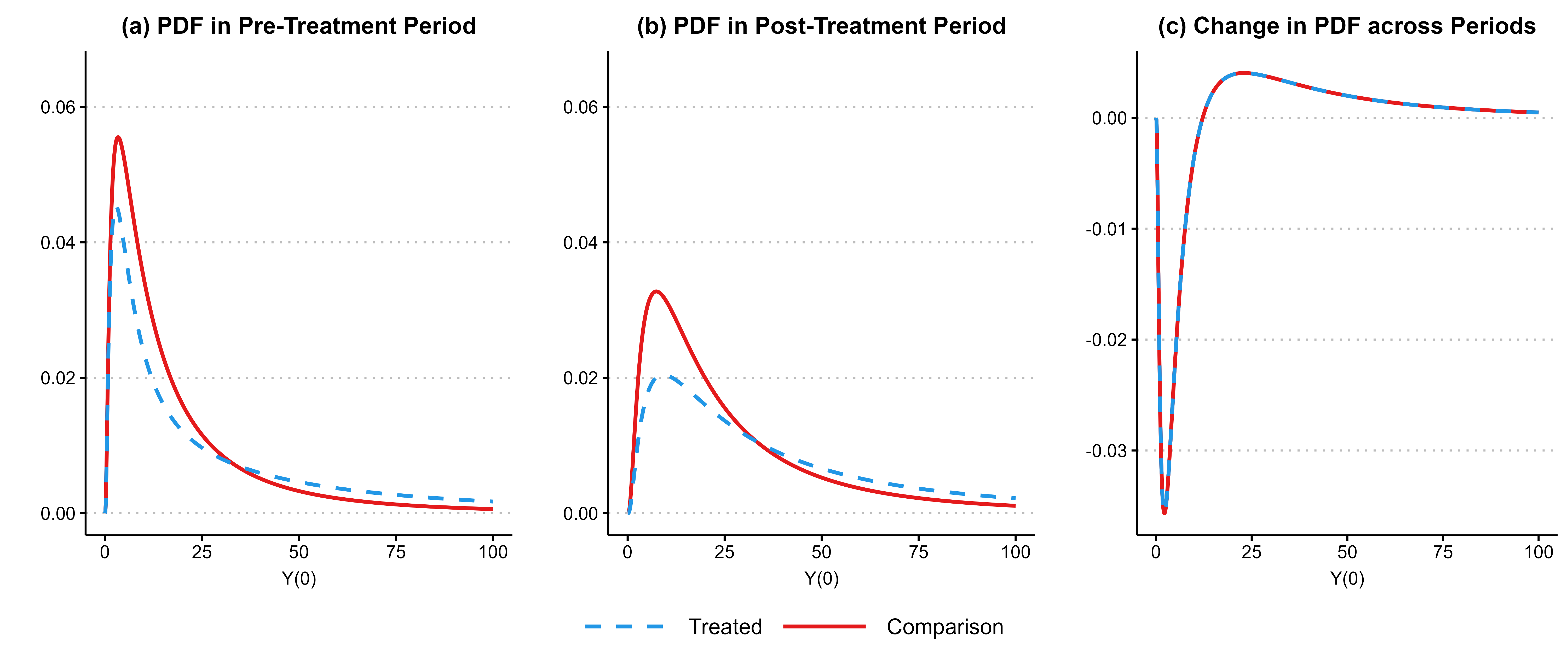}\\
    \caption{Illustration of Case 3}
    \label{fig:case3_example}
    \begin{minipage}{1\textwidth} 
    {\footnotesize Notes: Data generating process as discussed in Example \ref{ex:case3}.\par}
\end{minipage}
\end{figure}

\begin{table}[H]
\caption{Mean of $g(Y(0))$ by Group}
\label{tab:case3}
\centering
\begin{adjustbox}{ max width=1\linewidth, max totalheight=1\textheight}
\begin{threeparttable}
\begin{tabular}{@{}llrrr@{}} \toprule
g & Group      & Pre-treatment & Post-treatment & Change \\ \midrule
Levels           & Comparison & 22.65         & 33.12          & 10.47   \\
Levels           & Treated  & 51.10         & 61.57          & 10.47   \\
Log              & Comparison & 2.50          & 3.00           & 0.50    \\
Log              & Treated  & 3.00          & 3.50           & 0.50   \\
\bottomrule
\end{tabular}
\begin{tablenotes}[para,flushleft]
\small{
Notes: Data generating process as discussed in Example \ref{ex:case3}. }
\end{tablenotes}
\end{threeparttable}
\end{adjustbox}
\end{table}

\subsection{Relationship to Prior Work}

\begin{rem}[Empirical papers using parallel trends of CDFs] \label{rem: empirical papers}
Several empirical papers have used a DiD design to estimate the effects of a treatment on the distribution of an outcome. For example, \citet{almond_inside_2011} use DiD to estimate the effect of the food stamps program on the distribution of child birth weight. The validity of this analysis requires parallel trends throughout the birth weight distribution --- precisely the condition that Proposition \ref{prop: invariance to transformation did - superpop} shows is needed for the usual parallel trends assumption to be insensitive to functional form. Other recent papers (e.g. \citet{cengiz_effect_2019}, \citet{stepner_insurance_2019}) have conducted related distributional analyses. 
\end{rem}


\begin{rem}[Relationship to distributional DiD Models] \label{rem: cic}
Condition (\ref{eqn: parallel trends for cdfs - superpop}) implies that the full counterfactual distribution for the treated group, $Y_{i1}(0) | D_i=1$, is identified. In particular, by re-arranging terms in (\ref{eqn: parallel trends for cdfs - superpop}), we obtain that 
\begin{equation}
F_{Y_{i1}(0) | D_i=1}(y) = F_{Y_{i0}(0) | D_i=1}(y) + F_{Y_{i1}(0) | D_i=0}(y) - F_{Y_{i0}(0) | D_i=0}(y) \label{eqn: parallel trends cdf - rearranged 2},
\end{equation}
\noindent where the terms on the right-hand side correspond with CDFs of identified distributions. Condition (\ref{eqn: parallel trends for cdfs - superpop}) may thus be reminiscent of distributional DiD models such as \citet{AtheyImbens(06)}'s Changes-in-Changes (CiC) model, which infers the counterfactual distribution by assuming that the mapping between quantiles of $Y(0)$ for the treated and comparison populations remains stable over time, i.e.,
\begin{align}
F_{Y_{i1}(0) | D_i=1}(y) = F_{Y_{i0}(0)|D_i=1}( F_{Y_{i0}(0)|D_i=0}^{-1}( F_{Y_{i1}(0) | D_i=0} )(y)  )). \label{eqn: cic model}
\end{align}
The two ways of inferring the counterfactual distribution agree in Cases 1 and 2 above, but are generally non-nested otherwise. For instance, the CiC model does not hold in the example of Case 3 shown in Figure \ref{fig:case3_example} above. Conversely, one can construct examples where the CiC model holds when $Y(0)$ is normally distributed conditional on group and time period (with distinct means by group/period), whereas the parallel trends of CDFs assumption does not hold in this case as discussed in Example \ref{example: normal example}.
It is also straightforward to show that condition (\ref{eqn: parallel trends cdf - rearranged 2}) is non-nested with the distributional DiD models of \citet{Bonhomme_Sauder_2011} and \citet{callaway_quantile_2019}. 
\end{rem}

\subsection{Testable Implications of Invariance to Transformations}

We now show that condition (\ref{eqn: parallel trends for cdfs - superpop}) has testable implications, and thus can be rejected by the data. Recall that we can re-arrange the terms in (\ref{eqn: parallel trends for cdfs - superpop}) to obtain that 
\begin{equation}
F_{Y_{i1}(0) | D_i=1}(y) = F_{Y_{i0}(0) | D_i=1}(y) + F_{Y_{i1}(0) | D_i=0}(y) - F_{Y_{i0}(0) | D_i=0}(y). \label{eqn: parallel trends cdf - rearranged}
\end{equation}


The left-hand side of (\ref{eqn: parallel trends cdf - rearranged}) is a CDF and thus must be weakly increasing in $y$, but this is not guaranteed of the right-hand side. We can thus falsify condition (\ref{eqn: parallel trends cdf - rearranged}) if we reject the null that the right-hand side of (\ref{eqn: parallel trends cdf - rearranged}) is weakly increasing in $y$. This is in fact a \textit{sharp} testable implication, since the right-hand side of (\ref{eqn: parallel trends cdf - rearranged}) is guaranteed to be right-continuous and to have limits of 0 and 1 as $y\rightarrow \pm \infty$ from the properties of the CDFs in the linear combination.

\paragraph{Statistical Testing.} We now describe how one can conduct such tests in practice. For simplicity, we will focus on testing in the case where $y$ has finite support $\mathcal{Y}$, in which case the null is equivalent to testing that the implied distribution has non-negative mass at all support points. That is, we are interested in testing that
\begin{equation}f_{Y_{i1}(0)|D_i=1}(y) = f_{Y_{i0}(0)|D_i=1}(y) + f_{Y_{i1}(0)|D_i=0}(y) - f_{Y_{i0}(0)|D_i=0}(y) \geq 0 \text{ for all } y \in \mathcal{Y}, \label{eqn:densities under null} 
\end{equation}
\noindent where $f_{Y_{it}(0) | D_i = d}(y) = \expe{ 1[Y_{it}(0) = y] \,|\, D_i=d }$ is the probability mass function of $Y_{it}(0) | D_i=1$ at $y$. We can form an estimate of the implied PMF, $\hat{f}_{Y_{i1}(0)|D_i =1}(y)$, using sample analogs to the right-hand side of (\ref{eqn:densities under null}). The null of interest is then that $E[ \hat{f}_{Y_{i1}(0)|D_i =1}(y) ] \geq 0$ for all $y$, which can be tested using methods from the moment inequality literature (\citet{CanayShaikh(17)} provide a review). Note that researchers can also plot the implied distribution $\hat{f}_{Y_{i1}(0)|D_i=1}(y)$ to visualize potential violations of the ``parallel trends of distributions'' assumption. A similar approach could be taken in the case of continuous $Y$ using methods for testing a continuum of moment inequalities (or by discretizing the outcome).

\paragraph{Empirical Illustration.} We illustrate how such tests can be used with a stylized application to the effects of the minimum wage. Our pre-treatment period is 2007 or 2010 (depending on the specification), our post-treatment period is 2015, and the treatment is whether a state raised its minimum wage at any point between the pre-treatment and post-treatment periods. The outcome of interest is individual wages $W_i$ (where $W_i=0$ if $i$ is not working). We use data from \citet{cengiz_effect_2019}, who compile panel data on state-level minimum wages and employment-to-population ratios in 25-cent wage-bins.\footnote{Following \citet{cengiz_effect_2019}, wages are adjusted for inflation using the CPI-UR-S and expressed in constant 2016 dollars, and we exclude from the treated group states that only had small minimum wage changes of less than 25c (in 2016 dollars) or which affected less than 2 percent of the workforce.} Note that the employment-to-population ratio in wage bin $w$ corresponds with the mass function of $W_i$ at $w$. For each wage bin $w$, we infer the treated population's counterfactual employment-to-population ratio in wage bin $w$ as $\hat{f}_{Y_{i,post}(0)|D_i=1}(w) = \hat{f}_{Y_{i,pre}|D_i=1}(w) + (\hat{f}_{Y_{i,post}|D_i=0}(w) - \hat{f}_{Y_{i,pre} | D_i=0}(w))$, where $\hat{f}_{Y_{i,t}|D_i=d}$ is the sample employment-to-population ratio in period $t$ in states in treatment group $d$. In all calculations, we weight states by their population. 

\begin{figure}[!htb]
    \centering
    \includegraphics[width = 1\linewidth]{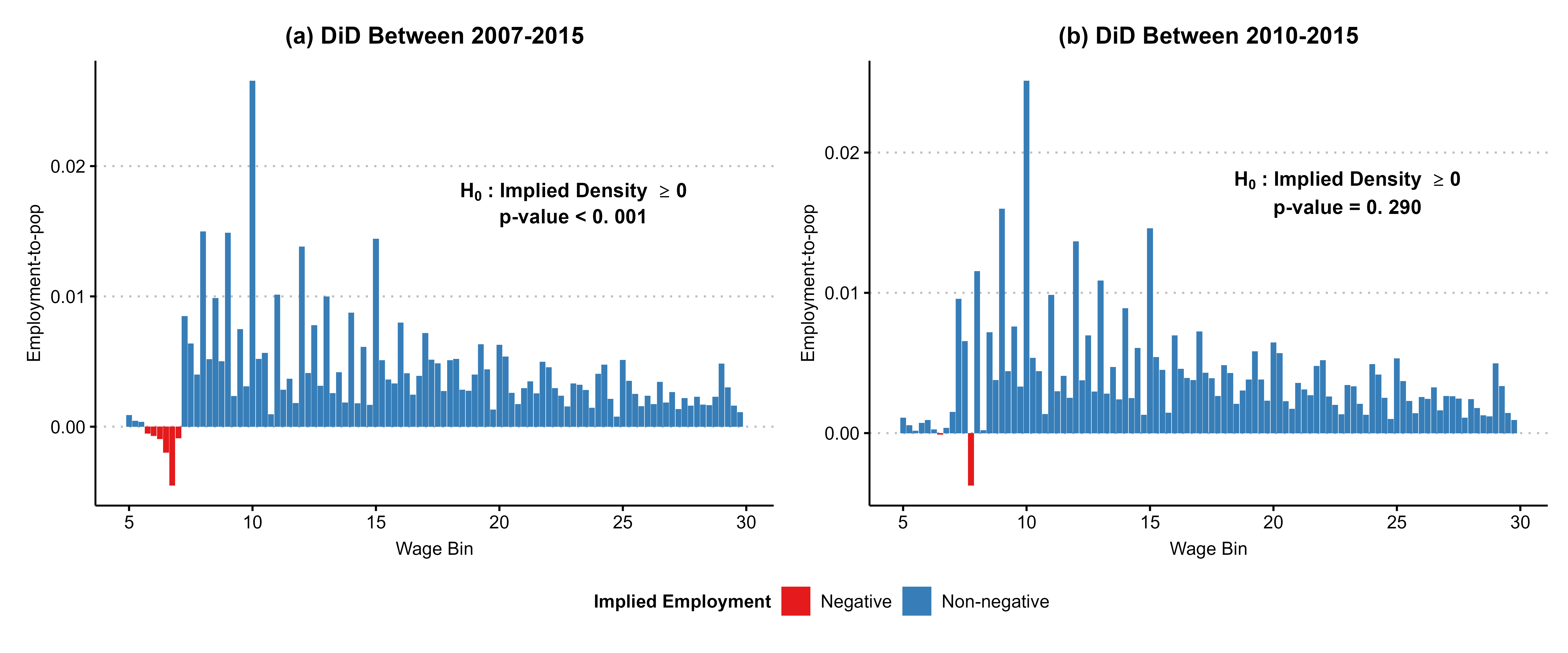}
    \caption{Implied Employment-to-Population Ratios for Treated States in 2015}
    \label{fig:mw-implied-densities}
\end{figure}

Figure \ref{fig:mw-implied-densities} shows the implied counterfactual densities $\hat{f}_{Y_{i,post}(0)|D_i=1}(w)$ under parallel trends of distributions. In the left panel, where the pre-treatment period is 2007, the figure shows that the implied density is \textit{negative} for wages between approximately \$5-7/hour. Intuitively, this occurs because the decrease in employment in such wage bins in comparison states between 2007 and 2015 is larger than the initial employment in treated states (who had lower baseline levels in these bins). One possible explanation for this pattern is that the increase in the federal minimum wage over this time period had a disproportionate impact on low-wage employment in states that did not have state-level minimum wages. To formally test the null that the implied density $f_{Y_{i,post}(0)|D_i=1}(w)$ is non-negative for all $w$, we estimate the variance-covariance matrix of the $\hat{f}_{Y_{i,post}(0)|D_i=1}(w)$ using a bootstrap at the state level, and then compare the minimum studentized value to a ``least-favorable'' critical value for moment inequalities that assumes all of the moments have mean 0 (see, e.g., Section 4.1.1 of \citet{CanayShaikh(17)}). Using such tests, we are able to reject the null hypothesis that all of the implied densities are positive ($p < 0.001$). 
This suggests that a researcher using such a DiD analysis to estimate the ATT for an outcome of the form $g(W_i)$ should be careful to justify the validity of the parallel trends assumption for the particular choice of functional form. 
By contrast, in the right panel of Figure \ref{fig:mw-implied-densities}, which shows results using the period 2010-2015, we see that the estimated counterfactual distribution has positive density nearly everywhere, and we cannot formally reject the hypothesis that it is positive everywhere ($p=0.29$). This does not necessarily imply that parallel trends holds for all transformations of the outcome, but insensitivity to functional form is not rejected by the data in this example.

\paragraph{Caveats.} We emphasize that failure to reject the null does not necessarily imply that parallel trends is insensitive to functional form, and relying on such tests may induce issues related to pre-testing; see \citet{roth_pre-test_2021} for a related discussion regarding testing for pre-existing trends. Nevertheless, such tests may be useful for detecting cases where it is clear from the data that parallel trends will be sensitive to functional form.

\subsection{Extensions}

\paragraph{Other classes of functions.} Following \citet{AtheyImbens(06)}, we focus on when parallel trends is invariant to all strictly monotonic transformations. One can show that parallel trends for all strictly monotonic transformations is in fact equivalent to parallel trends for all (measurable) transformations.\footnote{This is not the case for other assumptions --- e.g., the CiC model is invariant to strictly monotonic transformations but not to all measurable transformations.} In the working paper version of this paper, we showed that if one restricts to a smaller subset of transformations, then the counterfactual distribution for the treatment group may only be partially identified (see Appendix B of \citealp{roth2021when_WP}).

\paragraph{Other estimators.} We showed in the working paper version of this article that the $ATT$ is identified for all functional forms if and only if the full counterfactual distribution for the treated group, $F_{Y_{i1}(0)|D_i=1}(\cdot)$, is identified (see Proposition 4.1 of \citealp{roth2021when_WP}). This implies that to obtain any consistent estimator of the ATT (not just DiD), one must impose assumptions that either depend on functional form or that pin down the full counterfactual distribution of $Y(0)$ for the treated group.

\paragraph{Use of pre-treatment periods.} In settings with multiple pre-treatment periods, researchers may be inclined to use the absence of pre-existing trends as a justification for the validity of parallel trends for a particular functional form. However, as noted in \citet{kahn-lang_promise_2020}, pre-treatment parallel trends is neither necessary nor sufficient for post-treatment parallel trends (for a given functional form). Indeed, we showed in the working paper version of this article that if the support of $Y(0)$ is sufficiently rich, then there will be multiple transformations $g$ such that parallel trends holds in the pre-treatment period and a counterfactual post-treatment distribution such that it fails for at least one of these $g$. In finite samples, pre-trends tests may also have low power and relying on them can induce pre-test bias \citep{roth_pre-test_2021}.

\section{Implications for Applied Work}

Our results help to clarify the different paths a researcher can take
to justify point-identification of the ATT when considering using DiD or a related research design. First, the researcher can argue that treatment is as-if randomly assigned, which ensures identification of the average treatment effect for all functional forms. Second, the researcher can impose distributional assumptions that pin down the full counterfactual distribution of $Y_{i1}(0)$ for the treated group. This enables one to obtain an estimator that is valid for the ATT without additional functional form restrictions. For example, imposing parallel trends of CDFs ensures the validity of the DiD estimator for all functional forms. Finally, the researcher can use context-specific knowledge to argue for the validity of the parallel trends assumption for a particular functional form. A model with a Cobb-Douglas production function for $Y(0)$, for example, may imply parallel trends in logs but not levels. Our hope is that the results in this paper will be useful in providing researchers with a menu of options for more clearly delineating the justification for their research design.

\putbib
\end{bibunit}
\appendix
\begin{bibunit}

\section{Proof of Proposition \ref{prop: generation of parallel trends of cdfs}\label{sec: appendix proof of decomp of cdfs into mixtures}}
\begin{proof}
By Proposition \ref{prop: invariance to transformation did - superpop}, it suffices to show equivalence with (\ref{eqn: parallel trends for cdfs - superpop}). Observe that if (\ref{eqn: decomp of cdfs into mixtures}) holds, then both sides of (\ref{eqn: parallel trends for cdfs - superpop}) reduce to $\theta(G_{1}(y) - G_{0}(y))$, and so (\ref{eqn: decomp of cdfs into mixtures}) implies (\ref{eqn: parallel trends for cdfs - superpop}). To prove the converse, let $\mathcal{Y}$ denote the parameter space for $Y(0)$, and $\mathcal{Y}_y = \{ \tilde{y} \in \mathcal{Y} \,|\, \tilde{y} \leq y \}$. By assumption, we can write
$$F_{Y_{it}(0) | D_i = d}(y) = \int_{ \mathcal{Y}_y } f_{Y_{it}(0)| D_i=d} \, d \lambda,$$
\noindent where $\lambda$ is the dominating measure and $f_{Y_{it}(0)| D_i=d}$ is the density (the Radon-Nikodym derivative). It is immediate from the previous display that if (\ref{eqn: parallel trends for cdfs - superpop}) holds, then $f_{Y_{i1}(0)|D_i=1} - f_{Y_{i0}(0)|D_i=1} = f_{Y_{i1}(0)|D_i=0} - f_{Y_{i0}(0)|D_i=0} $, $\lambda \text{ a.e.}$ The desired result then follows by applying Lemma \ref{lem: helper lemma for decomp proof} to the CDFs on both sides of (\ref{eqn: parallel trends for cdfs - superpop}). 
\end{proof}

\begin{lem} \label{lem: helper lemma for decomp proof}
Suppose the CDFs $F_1$ and $F_2$ are such that $F_j(y) = \int_{\mathcal{Y}_y} f_j d\lambda$. Then we can decompose $F_j(y)$ as
\begin{equation}
F_j(y) = (1-\theta) F_{min}(y) + \theta \tilde{F}_j(y), \label{eqn: decomp of cdfs in proof}    
\end{equation}
where $F_{min}$ and $\tilde{F}_1, \tilde{F}_2$ are CDFs, $\theta \in [0,1]$, and $\theta$ and $\tilde{F}_j$ depend on $f_1$ and $f_2$ only through $f_1 - f_2$. 
\end{lem} 
\begin{proof}

To prove the claim, set $\theta =1- \int_{\mathcal{Y}} \min\{f_1, f_2 \} d\lambda$. It is immediate that $\theta \in [0,1]$. Suppose first that $\theta \in (0,1)$. Define
\begin{align*}
&f_{min} = \dfrac{ \min\{f_1, f_2 \} }{ \int_{\mathcal{Y}} \min\{f_1, f_2 \} d\lambda } =  \dfrac{ \min\{f_1, f_2 \} }{1-\theta}
\end{align*}
\noindent and
\begin{align*}
&\tilde{f}_j(x) = \dfrac{ f_j - \min\{f_1, f_2 \} }{ \int_{\mathcal{Y}} (f_j - \min\{f_1, f_2 \}) d\lambda } = \dfrac{ f_j - \min\{f_1, f_2 \} }{\theta } \text{ for } j = 1,2.     
\end{align*}
\noindent By construction, $f_{min}$ and the $\tilde{f}_j$ integrate to 1 and are non-negative, so that $F_{min}(y) = \int_{\mathcal{Y}_y} f_{min} d\lambda$ and $\tilde{F}_j(y) = \int_{\mathcal{Y}_y} \tilde{f}_j d\lambda$ are valid CDFs. Moreover, $f_j = (1-\theta) f_{min} + \theta \tilde{f}_j$ by construction, so that (\ref{eqn: decomp of cdfs in proof}) holds. Finally, observe that $\min\{f_1,f_2\} = f_1 - (f_1 - f_2)_+$, where $(a)_+$ denotes the positive part of $a$. It follows that $\theta = 1-\int_{\mathcal{Y}}(f_1 - (f_1 - f_2)_+) d\lambda = \int_{\mathcal{Y}} (f_1 - f_2)_+ d\lambda$, which depends only on $f_1 - f_2$. (In fact, note that $\int_{\mathcal{Y}} (f_1 - f_2)_+ d\lambda = \frac{1}{2} \int_{\mathcal{Y}} |(f_1 - f_2)| d\lambda$, and thus $\theta$ is the total variation distance between $f_1$ and $f_2$.) Likewise, $\tilde{f}_1 = (f_1 - f_2)_+ / \theta$ and $\tilde{f}_2 = (f_2 - f_1)_+ / \theta$, and so depend only on $f_1 - f_2$. This completes the proof for the case where $\theta \in (0,1)$. If $\theta = 0$, then $F_1(y) = F_2(y)$ and so the claim holds trivially with $F_{min}(y) = F_1(y) = F_2(y)$ and $\tilde{F}_j(y)$ arbitrary. If $\theta = 1$, then $\min\{ f_1, f_2\} = 0 \; \lambda \text{ a.e}$, and so $f_1 = (f_1 - f_2)_+ \; \lambda \text{ a.e.},$ and $f_2 = (f_2 - f_1)_+ \; \lambda \text{ a.e.}$. Thus, the claim holds trivially with $\tilde{f}_j = f_j$, $\tilde{F}_j(y) = \int_{\mathcal{Y}_y} \tilde{f}_j d\lambda$, and $F_{min}(y)$ arbitrary.   
\end{proof}

\end{bibunit}
\end{document}